\def\changesHilighted{false}

\documentclass[conference]{IEEEtran}
\IEEEoverridecommandlockouts
\usepackage[letterpaper, top=0.75in,bottom=0.75in,left=0.75in,right=0.75in]{geometry}
\usepackage{cite}
\usepackage{amsmath,amssymb,amsfonts}
\usepackage{graphicx}
\usepackage{textcomp}
\usepackage{xcolor}
\usepackage{color}
\usepackage{CJK}    
\usepackage{subfigure}
\usepackage[normalem]{ulem}
\usepackage{caption}
\usepackage{mathtools} 
\mathtoolsset{showonlyrefs}
\usepackage{amsthm}
\usepackage{float}
\usepackage{textcomp}

\allowdisplaybreaks

\newtheorem{theorem}{Theorem}
\newtheorem{lemma}[theorem]{Lemma}
\newtheorem{proposition}[theorem]{Proposition}

\theoremstyle{definition}

\newtheorem{problem}{Problem}

\theoremstyle{remark}
\newtheorem{remark}{Remark}

\definecolor{dogwoodrose}{rgb}{0.84, 0.09, 0.41}
\definecolor{green}{rgb}{0, 0.39, 0}
\newcommand{\add}[1]{\textcolor{blue}{#1}}
\newcommand{\masaki}[1]{\textcolor{dogwoodrose}{{\textbf{[Masaki: #1]}}}}
\definecolor{slateblue}{rgb}{0.42, 0.35, 0.8}
\newcommand{\aiyi}[1]{\textcolor{slateblue}{{\textbf{[Aiyi: #1]}}}}
\newcommand{\deng}[1]{\textcolor{green}{{\textbf{[DENG: #1]}}}}
\newcommand{\del}[1]{\textcolor{red}{\sout{#1}}}

\newcommand{\norm}[1]{\lVert #1 \rVert}
\newcommand{\abs}[1]{\left| #1 \right|}

\ifnum\pdfstrcmp{\changesHilighted}{false}=0
    \renewcommand{\add}[1]{#1}
    \renewcommand{\masaki}[1]{}
    \renewcommand{\del}[1]{}
    \renewcommand{\deng}[1]{}
    \renewcommand{\aiyi}[1]{}
\fi

\usepackage{balance}

\begin{document}

\title{Collision-Free Shepherding Control of a Single Target within a Swarm
\thanks{This work was supported in part by JSPS KAKENHI Grant Number JP21H01352.}
}

\author{\IEEEauthorblockN{Yaosheng Deng,
Aiyi Li, Masaki Ogura, and
Naoki Wakamiya}
\IEEEauthorblockA{Graduate School of Information Science and Technology,
Osaka University,\\
Suita, Osaka 565-0871, Japan\\
\{ys-deng, li-aiyi, m-ogura, wakamiya\}@ist.osaka-u.ac.jp}}

\maketitle

\begin{abstract}

The shepherding problem refers to guiding a group of agents (called sheep) to a specific destination using an external agent with repulsive forces (called shepherd). Although various movement algorithms for the shepherd have been explored in the literature, there is a scarcity of methodologies for \emph{selective guidance}, which is a key technology for precise swarm control. Therefore, this study investigates the problem of guiding a single target sheep within a swarm to a given destination using a shepherd. We first present our model of the dynamics of sheep agents and the interaction between sheep and shepherd agents. The model is shown to be well-defined with no collision if the interaction magnitude between sheep and shepherd is reasonably limited. Based on the analysis with Lyapunov stability principles, we design a shepherd control law to guide the target sheep to the origin while avoiding collisions among sheep agents. Experimental results demonstrate the effectiveness of the proposed method in guiding the target sheep in both small and large scale swarms. 
\end{abstract}


\section{Introduction}

It has long been observed that various biological individuals tend to exhibit swarming behavior. Examples of swarms include flocks of birds, schools of fish, herds of animals, and colonies of bacteria~\cite{budrene1995dynamics,vicsek1995novel}, and engineering applications include the formation control of multi-robot teams and autonomous aircraft. Many studies~\cite{Gazi,goel,intro-unbounded} have investigated the relevant properties of swarm motion based on the attraction-repulsion swarm model and its applications in areas such as UAV swarms and multi-robot systems.

As research on swarm intelligence advances, there has been increasing attention to the use of an external special agent (called shepherd) to control the movement of swarm agents (called sheep). This process is referred to as \textit{shepherding}~\cite{long2020comprehensive}. The shepherding model specifically refers to the situation in which ``sheep'' agents avoid a ``shepherd'' agent while also interacting with other sheep, according to a swarm model. In this context, we can find various methods~\cite{shepherding-swarm2, shepherding-swarm3, shepherding-swarm4} to guide a whole swarm of sheep agents. 

There is emerging interest in guiding a part of the swarm (e.g., a target sheep) to the goal because of its potential engineering applications in the context of micro-/nano-robotic swarms~\cite{engineeringvalue1,engineeringvalue2}. For example, Deptula et al.~\cite{Deptula} proposed a control algorithm for a shepherd to guide a targeted sheep agent towards its destination. Licitra et al.~\cite{Licitra} designed a controller for herding one sheep within a swarm to its destination without considering swarm interaction forces in the dynamics model. Using the framework of reinforcement learning, Sebasti\'an and Montijano~\cite{onlyone1} developed a method for guiding a few agents within a swarm with heterogeneous agent dynamics. Le et al.~\cite{onlyone2} proposed a method to use a shepherd to regulate a sheep within a small swarm. Zhang et al.~\cite{onlyone3} designed an algorithm for capturing and regulating a portion of a swarm using multiple shepherds. 

However, the shepherd control methods mentioned above~\cite{Deptula,Licitra,onlyone1,onlyone2,onlyone3} do not consider collisions within the swarm; their algorithms assume that no pairs of agents in the swarm collide. Because the occurrence of a collision can disrupt an entire system and lead to failed tasks~\cite{collision5,collision_otherswarm1,bandyopadhyay2021detection}, it is practically important to develop a shepherding control method with guaranteed collision-free properties. Although several studies have proposed distributed~\cite{collision_otherswarm3,collision_otherswarm4} and leader-based~\cite{collision_otherswarm2} control methods for swarm systems with collision-free guarantees, these studies have not been directly extended to the context of shepherding control. Therefore, to design a stable shepherd controller to achieve certain control objectives, the first consideration must be to carefully analyze the collision properties of shepherding swarms.

In this study, the control objective is to use a shepherd agent to guide a sheep agent (target) in a swarm to the origin (goal). First, the repulsion-bounded swarm model is described. We then carefully analyze and rigorously prove that under relatively mild conditions, sheep in the swarm do not collide, even under the influence of a shepherd. Based on this analysis, we use the Lyapunov stability principle to design a motion controller for the shepherd to guide the target to the endpoint (i.e., the origin). The effectiveness of the proposed method is illustrated by numerical simulations, in which we compare it with a baseline strategy based on the farthest agent-targeting strategy presented in~\cite{tsunoda2018analysis}. 

This paper is organized as follows: Section~\ref{Problem statement} describes the shepherding swarm model and the problem for controlling one target to the origin; in Section~\ref{section3} we analyze the collision-free property of the shepherding swarm model; in Section~\ref{Control Design} we design a shepherd motion controller to guide a specific target to the destination. Section~\ref{Numerical Simulation} provides several numerical simulations that highlight the shepherding behavior. 


\section{Problem statement}\label{Problem statement}

We consider a swarm of~$N>2$ individuals (\textit{sheep}) and one herder (\textit{shepherd}) in~$\mathbb{R}^2$. The positions of sheep~$i$ and shepherd at time $t\geq 0$ are denoted by $x_i(t)\in \mathbb{R}^2$ and~$y(t)\in \mathbb{R}^2$, respectively. Let $\mathcal{N}=\{1,2,...,N\}$.

We assume that, for all $i \in \mathcal N$, the velocity~$\dot x_i$ of the $i$th sheep is specified by a linear combination of three components: 
\begin{equation} \label{dot xi}
\dot{x}_i = f_{ai} + f_{bi} + f_{yi}, 
\end{equation}
where $f_{ai}$ and~$f_{bi}$ represent the attraction and repulsion functions from the other sheep to the $i$th sheep, respectively, and~$f_{yi}$ represents the function from the shepherd to the $i$th sheep. 

Below, we describe how  $f_{ai}$, $f_{bi}$, and~$f_{yi}$ are constructed. First, the attraction function of the $i$th sheep is given by:
\begin{equation}\label{fai}
    f_{ai} = \sum_{j\in \mathcal{N}{\backslash \{i\}}}\phi_a(x_i-x_j), 
\end{equation} 
where the function~$\phi_a\colon \mathbb R^2\to\mathbb R^2$ is defined as 
\begin{equation}\label{phia}
    \phi_{a}({x}) = -m_ax
\end{equation}
for a positive constant~$m_a$. Then, 
the bounded repulsion function of the $i$th sheep is given by
\begin{equation}\label{fbi_bounded}
    f_{bi} = \sum_{j\in \mathcal{N}{\backslash \{i\}}}\phi_b(x_i-x_j),
\end{equation} 
where the function~$\phi_b\colon \mathbb R^2\to\mathbb R^2$ is defined as 
\begin{equation}\label{phi_b_bounded}
    \phi_{b}({x}) = 
    \begin{cases}
        m_b\dfrac{{x}}{\norm{{x}}^{2}},  & \mbox{if }\norm{{x}}>\ell_b,\vspace{2mm}\\
M_b \dfrac{{x}}{\norm{{x}}}, & \mbox{if }\norm{{x}}\leq \ell_b,
    \end{cases}
\end{equation}
for positive constants $\ell_b$, $m_b$, and~$M_b$ satisfying 
\begin{equation}\label{eq:MB=}
    M_b = \frac{m_b}{\ell_b}, \quad \ell_b<R.
\end{equation}
Finally, we define $f_{yi}$ as 
\begin{equation}\label{fyi}
  f_{yi} = \begin{cases}
\gamma_1\dfrac{x_i-y}{\norm{x_i-y}^{2}},\vspace{2mm}&
\mbox{if $\norm{x_i-y}>\ell_y$,}
\\
\displaystyle
\gamma_2 \dfrac{x_i-y}{\norm{x_i-y}^{2}}g(\norm{x_i-y}),&
\mbox{if $\norm{x_i-y}\leq \ell_y$,} 
  \end{cases}  
\end{equation}
where the function~$g\colon [0, \infty) \to [0, \infty)$ is defined by:
\begin{equation}\label{gaussian}
    g({d})=
    \begin{cases}
    0,&\mbox{if $d=0$,}
    \\
        \exp\left({-{1}/{d}}\right),&\mbox{otherwise}
    \end{cases}
\end{equation}
and the parameters $\gamma_1$, $\gamma_2$, and~$\ell_y$ are positive constants that satisfy
\begin{equation}\label{gamma_epsilon}
    \gamma_2 = \gamma_1 g^{-1}(\ell_y).
\end{equation}

We can now state the problem studied in this paper. 

\begin{problem}\label{prb:}
Assume that the shepherd knows all sheep positions and can calculate the interaction force~$f_{aT}+f_{bT}$ of the target sheep~$T\in \mathcal N$.
Design a controller for the velocity of the shepherd such that the target sheep is asymptotically guided to the origin\del{, i.e., $x_T(t)$ tends to the origin as $t\to\infty$}.
\end{problem}

\section{Collision-Free Condition}\label{section3}

In this section, we present a sufficient condition that guarantees that no pairs of sheep will collide within a swarm. We begin by presenting the main results of this section. Within the theorem, the maximum shepherd force for one sheep is defined as follows: 
\begin{equation}
f_y^{\max}=\frac{\gamma_1}{\ell_y},
\end{equation}
which plays an important role in this process.

\begin{theorem}\label{noncollision}
Suppose that
\begin{equation}
    x_i(0) \neq x_j(0)
\end{equation}
for all distinct pairs $(i, j) \in \mathcal N\times \mathcal N$. If 
\begin{equation}\label{M_b-m_a>fy}
f_y^{\max}<M_b,
\end{equation}
then 
\begin{equation}\label{eq:nocollision}
x_i(t)\ne x_j(t)
\end{equation}
for all distinct pairs $(i, j) \in \mathcal N\times \mathcal N$ and~$t>0$. 
\end{theorem}

In the remainder of this section, we present a proof of Theorem~\ref{noncollision} using the following notation.
For all $i, j\in \mathcal N$, let 
\begin{equation}
\delta_{ij}=\norm{x_i-x_j}.
\end{equation}
The first collision time of the system is defined as follows:
\begin{equation}
 t_1 =  \min\{ t\geq 0 \mid 
 \min_{i,j\in\mathcal{N},\,i\ne j}\delta_{ij}(t) = 0
 \}.
\end{equation}
Set~$\mathcal{C}$ is defined as 
\begin{equation}
\begin{aligned}
    \mathcal{C}=
    \begin{cases}
        \emptyset, &\mbox{if $t_1= \infty$,}
        \\
        \{i \in \mathcal N 
    \mid \delta_{ij}(t_1) = 0 \mbox{ for some } j\neq i\}, &\mbox{otherwise.}
    \end{cases}
\end{aligned}
\end{equation}
We then define the complement set as 
\begin{equation}
    \mathcal{R}= \mathcal{N} \setminus \mathcal{C}.
\end{equation}

We prove Theorem~\ref{noncollision} by contradiction. First, the proof is outlined. Let us introduce the nonnegative function 
\begin{equation}\label{Xc*}
\mathcal{X}_2=\sum_{i,j\in\mathcal{C},i\ne j}\delta_{ij}^2.
\end{equation}
If a finite-time collision occurs, that is, if the conclusion of Theorem~\ref{noncollision} is violated, then $t_1$ is finite; therefore, we have $\mathcal X_2(t_1) = 0$. Conversely, as will be shown later, a careful analysis of the derivative of~$\mathcal X_2(t)$ allows us to show that function~$\mathcal X_2(t)$ increases on the interval $[t_1-\tau, t_1)$ for a sufficiently small $\tau > 0$ under inequality~\eqref{M_b-m_a>fy}. This observation leads to $\mathcal X_2(t)$ being negative on the interval, which contradicts the intrinsic nonnegativity of function~$\mathcal X_2$.

To proceed further, we first evaluate the derivative of~$\mathcal X_2$. 
A straightforward calculation shows that the derivative can be decomposed as follows:
\begin{equation}\label{eq:dotX2}
\dot{\mathcal{X}}_2=
 \mathcal{I}_{1} + \mathcal{I}_{2}+\mathcal{I}_{3}+ \mathcal{I}_{4}+\mathcal{I}_{5},
\end{equation}
where
\begin{IEEEeqnarray}{rCl}\label{i+i+i+i+i} 
\mathcal{I}_{1}&=&2\sum_{\begin{smallmatrix}i,j\in\mathcal{C}\\ i\neq j\end{smallmatrix}}(x_i-x_j)^{\top}(f_{yi}-f_{yj}),
\\
\mathcal{I}_{2}&=&2\sum_{\begin{smallmatrix}i,j,k\in\mathcal{C}\\ i\neq j, j\neq k\end{smallmatrix}}(x_i-x_j)^{\top}\left(\phi_a(x_i-x_k)+\phi_b(x_i-x_k)\right), 
\nonumber\\
\mathcal{I}_{3}&=&2\sum_{\begin{smallmatrix}i,j,k\in\mathcal{C}\\ i\neq j, j\neq k\end{smallmatrix}}(x_j-x_i)^{\top}\left(\phi_a(x_j-x_k)+\phi_b(x_j-x_k)\right), 
\nonumber\\
\mathcal{I}_{4}&=&2\sum_{\begin{smallmatrix}i,j\in\mathcal{C}\\ k\in \mathcal{R}, i\neq j\end{smallmatrix}}(x_i-x_j)^{\top}\left(\phi_a(x_i-x_k)+\phi_b(x_i-x_k)\right), 
\nonumber\\
\mathcal{I}_{5}&=&2\sum_{\begin{smallmatrix}i,j\in\mathcal{C}\\k\in \mathcal{R},i\neq j\end{smallmatrix}}(x_j-x_i)^{\top}\left(\phi_a(x_j-x_k)+\phi_b(x_j-x_k)\right)\nonumber.
\end{IEEEeqnarray}
Because the symmetry~\cite{syha} of these expressions with respect to indices $i$ and~$j$ allows us to show 
$\mathcal{I}_{2}=\mathcal{I}_{3}$
and $\mathcal{I}_{4}=\mathcal{I}_{5}$, 
we obtain 
\begin{equation}\label{x2=i1+2i2+2i4}
\dot{\mathcal{X}}_2=
 \mathcal{I}_{1} + 2\mathcal{I}_{2}+2\mathcal{I}_{4}.
\end{equation}

To further evaluate the derivative~$\dot{\mathcal X}_2$, we provide lower bounds for~$\mathcal I_2$ and~$\mathcal I_4$. A lower bound for~$\mathcal I_2$ is presented in the following proposition: 

\begin{proposition}\label{prop_omitted}
Let $t\geq 0$ be arbitrary. 
If 
\begin{equation}
\delta_{\mathcal C, \max}(t) = \max_{i,j\in\mathcal{C}, i\ne j}\delta_{ij}(t)
\end{equation}
satisfies 
\begin{equation}\label{eq:keyineq}
\delta_{\mathcal C, \max} (t) \leq \ell_b, 
\end{equation}
then 
\begin{equation}\label{i2}
\mathcal{I}_{2}(t) \geq  {M_b\abs{\mathcal{C}}}\mathcal{X}_1(t)-{m_a\abs{\mathcal{C}}}\mathcal{X}_1^2(t), 
\end{equation}
where $\mathcal X_1$ is defined as
\begin{equation}\label{Xc}
\mathcal{X}_1=\sum_{i,j\in\mathcal{C},i\ne j}\delta_{ij}. 
\end{equation}
\end{proposition}

\begin{proof}
{We omit the time variable~$t$ for the simplicity of notations within the proof.}
If $\mathcal C=\emptyset$, the inequality holds vacuously. We assume that $\mathcal C$ is nonempty. By using \eqref{phia}, \eqref{phi_b_bounded}, and~\eqref{i2}, we can show that 
\begin{equation}
\begin{multlined}
\mathcal{I}_{2}
=\!2\sum_{\begin{smallmatrix}i,j,k\in\mathcal{C}\\ i\neq j\neq k\end{smallmatrix}}(x_i-x_j)^{\top}\Bigg(m_a(x_k-x_i)-M_b\frac{x_k-x_i}{\norm{x_k-x_i}}\Bigg). 
\end{multlined}
\end{equation}
By exchanging indices $i$ and~$k$, we obtain 
\begin{equation}
\begin{multlined}
\mathcal{I}_{2}
=2\!\sum_{\begin{smallmatrix}i,j,k\in\mathcal{C}\\ i\neq j\neq k\end{smallmatrix}}(x_k-x_j)^{\top}\Bigg(m_a(x_i-x_k)-M_b\frac{x_i-x_k}{\norm{x_k-x_i}}\Bigg).
\end{multlined}
\end{equation}
Therefore, 
\begin{equation}
\begin{multlined}
\mathcal{I}_{2}
=\sum_{\begin{smallmatrix}i,j,k\in\mathcal{C}\\ i\neq j\neq k\end{smallmatrix}}(x_i-x_j)^{\top}\Bigg(m_a(x_k-x_i)-M_b\frac{x_k-x_i}{\norm{x_k-x_i}}\Bigg)\\
+ \sum_{\begin{smallmatrix}i,j,k\in\mathcal{C}\\ i\neq j\neq k\end{smallmatrix}}(x_k-x_j)^{\top}\Bigg(m_a(x_i-x_k)-M_b\frac{x_i-x_k}{\norm{x_k-x_i}}\Bigg). 
\end{multlined}
\end{equation}
Then, we can further simplify $\mathcal{I}_{2}$ as 
\begin{equation}
\begin{aligned}
\mathcal{I}_{2}
&=\sum_{\begin{smallmatrix}i,j,k\in\mathcal{C}\\ i\neq j\neq k\end{smallmatrix}}\bigl(-m_a\norm{x_k-x_i}^2+M_b\norm{x_k-x_i}\bigr)\\
& = 
{{M_b\abs{\mathcal{C}}}\mathcal{X}_1(t)-{m_a\abs{\mathcal{C}}}\mathcal{X}_2(t)},
\end{aligned}
\end{equation}
\aiyi{The above formulas are maybe too detailed.}\masaki{Agreed}\deng{If we do not exceed the 8 pages limitation, I would keep the equations as detailed as possible.}\aiyi{Ok, we can remove it from the camera-ready version.}\masaki{Agreed}which proves the desired inequality~\eqref{i2} using the trivial inequality $\mathcal X_2 \leq \mathcal X_1^2$. 
\end{proof}

Then, we derive a lower bound for~$\mathcal I_4$. We begin by presenting the following technical lemma. The proof of the lemma~\ref{IaIb} is omitted due
to space limitations.

\begin{lemma}\label{IaIb}
Let $t\geq 0$ and $i,j\in\mathcal{C}$ be arbitrary. Assume $i\ne j$. Define 
\begin{equation}\label{Ia}
\begin{multlined}[.85\linewidth]
\mathcal{I}_a(t) =(x_i(t)-x_j(t))^{\top}\phi_a(x_i(t)-x_{k}(t)) 
\\+ (x_j(t)-x_i(t))^{\top}\phi_a(x_j(t)-x_{k}(t))
\end{multlined}
\end{equation}
and 
\begin{equation}\label{ib}
\begin{multlined}[.85\linewidth]
\mathcal{I}_{b}(t) =-\frac{1}{2}\Big((x_i(t)-x_j(t))^{\top}\phi_b(x_{k}(t)-x_i(t))
\\+(x_j(t)-x_i(t))^{\top}\phi_b(x_{k}(t)-x_j(t))\Big). 
\end{multlined}
\end{equation}
Then, we have 
\begin{equation}\label{ia>=} \mathcal{I}_a(t)= -\frac{m_a}{2} \delta_{ij}^2(t)
\end{equation}
and 
\begin{equation}\label{ib>=}
\begin{multlined}
      \mathcal{I}_{b}(t)  \geq \\\begin{cases}
0,\mbox{ if $\norm{x_k(t)-x_i(t)}<\ell_b$ and~$\norm{x_k(t)-x_j(t)}<\ell_b$,}\!\! 
\\
\displaystyle
- \frac{m_b}{2}\left(\frac{1}{\ell_b}-\frac{1}{\ell_b+\delta_{\mathcal C, \max}(t)}\right) \delta_{ij}(t),
\\\quad\mbox{if $\norm{x_k(t)-x_i(t)}\geq\ell_b$ and 
$\norm{x_k(t)-x_j(t)}\geq\ell_b$,}
\\
\displaystyle
- \frac{m_b}{2}\left(\frac{1}{\ell_b}-\frac{1}{\ell_b+\delta_{\mathcal C, \max}(t)}\right) \delta_{ij}(t),
\\\quad\mbox{if $\norm{x_k(t)-x_i(t)}\geq\ell_b$ and~$\norm{x_k(t)-x_j(t)}\leq \ell_b$.} 
  \end{cases}  
  \end{multlined}
\end{equation}
\end{lemma}


Using Lemma~\ref{IaIb}, we can prove the following lower bound for~$\mathcal I_4$. The proof of the proposition~\ref{lemma_infi4} is omitted due
to space limitations.

\begin{proposition}\label{lemma_infi4}
Let $t\geq 0$ be arbitrary.
If inequality~\eqref{eq:keyineq} holds, then 
\begin{equation}
\label{i4>=}
\mathcal{I}_{4}(t)\geq 
-\lvert \mathcal R\rvert\left({M_b}-\frac{M_b\ell_b}{\ell_b+\delta_{\mathcal C, \max}(t)}+m_a\mathcal{X}_1(t)\right)\mathcal{X}_1(t). 
\end{equation}
\end{proposition}

We are now ready to prove Theorem~\ref{noncollision}.

\begin{proof}[Proof of Theorem~\ref{noncollision}]
Assume that inequality~\eqref{M_b-m_a>fy} holds; i.e., $f_y^{\max} < M_b$. Under this assumption, we need to show that set~$\mathcal C$ is empty. Let us assume the contrary to derive a contradiction. Let us assume that $\mathcal C$ is nonempty. Then, we have 
\begin{equation}\label{eq:C>2}
    \lvert \mathcal C \rvert \geq 2.
\end{equation}
Further, $t_1$ is finite and we have
$\lim_{t\xrightarrow{}t_1^-}\mathcal{X}_1(t)=0$
and
 \begin{equation}
\lim_{t\xrightarrow{}t_1^-}\left(1-\frac{\ell_b}{\ell_b+\delta_{\mathcal C, \max}(t)}\right)=0. 
 \end{equation}
Hence, for~$\epsilon =({M_b - f_y^{\max}})/4>0$, there exists $\tau >0$ such that, if $t\in (t_1-\tau,t_1)$, then 
\begin{equation}\label{eq:<epsilon}
\begin{gathered}
m_a\abs{\mathcal{R}}\mathcal{X}_1<\epsilon,\quad  m_a\abs{\mathcal{C}}\mathcal{X}_1<\epsilon,
\\
M_b\abs{\mathcal{R}}\left(1-\frac{\ell_b}{\ell_b+\delta_{\mathcal C, \max}(t)}\right)<\epsilon.
\end{gathered}
\end{equation}
Now, from \eqref{i+i+i+i+i}, \eqref{x2=i1+2i2+2i4}, \eqref{i2}, \eqref{i4>=}, \eqref{eq:C>2}, and \eqref{eq:<epsilon}, we can evaluate the derivative~$\dot{\mathcal X}_2(t)$ for~$t\in (t_1-\tau, t_1)$ as 
\begin{equation}\label{dxc2/dt>_delta}
\begin{aligned}
\dot{\mathcal{X}}_2(t)&
\geq
\begin{multlined}[t][.85\linewidth]
    {2\abs{\mathcal{C}}}\left(M_b\mathcal{X}_1(t)-m_a\mathcal{X}_1^2(t)\right) -2\mathcal{X}_1(t)f_y^{\max}\\
 \!\!\!\!- \Bigg(M_b-\frac{M_b\ell_b}{\ell_b+\delta_{\mathcal C, \max}(t)}+m_a\mathcal{X}_1(t)\Bigg){2\abs{\mathcal{R}}}\mathcal{X}_1(t)\hspace{14mm}
\end{multlined}\\
&\geq 2\mathcal{X}_1\left(2M_b-2f_y^{\max}-3\epsilon\right)>0.
\end{aligned}
\end{equation}
This inequality implies $\mathcal{X}_2(t_1-\tau^*) < \mathcal X_2(t_1)= 0$, 
which contradicts the intrinsic nonnegativity of~$\mathcal X_2$, as desired. 
\end{proof}

\begin{remark}
Our analysis in Theorem~\ref{noncollision} is based on the implicit assumption that when the first collision occurs, another collision by other sheep does not occur simultaneously at another location. Although it is possible to relax this assumption, we chose not to present the analysis due to space limitations. 
\end{remark}

\section{Control Design}\label{Control Design}

From Theorem~\ref{noncollision}, we conclude that no collision occurs, the dynamics of the system are continuous, and we do not need to consider collisions when using the model~\eqref{dot xi} to design a controller for regulating the target to the origin. In this section, we describe the design of the controller for the shepherd to ensure that the target sheep is controlled to the origin. We specifically use the Lyapunov stability theorem to ensure that the shepherd has ability of locally regulating the target to the origin. 

Let us denote the sum of the interaction forces of the target as 
\begin{equation}\label{wt}
w_T = f_{aT} + f_{bT}.
\end{equation}
We remark that this quantity is supposed to be available by the shepherd. Then, using \eqref{fyi}, we can rewrite the velocity of the target as
\begin{equation}\label{dotxt}
\begin{multlined}
\dot{x}_T = \\
    \begin{cases}
        \gamma_1\norm{x_T-y}^{-2}(x_T-y) + w_T,\ \ \ \ \ \ \ \!\mbox{if }\norm{x_T-y}>\ell_y, \\
\gamma_2 \norm{x_T-y}^{-2}g(\norm{x_T-y})(x_T-y)+ w_T,\  \mbox{otherwise.}
    \end{cases}
    \end{multlined}
\end{equation}
\add{The system \eqref{dotxt} is controllable when $f_y^{\max}\geq w_T$ for $t>0$. The controller design should base on the controllability of the \eqref{dotxt}. Thus we select the gains that satisfy
\begin{equation}
f_y^{\max}<M_b \quad f_y^{\max}\geq N(M_b-m_a)
\end{equation}
}
To realize our control objective, we set the desired position of the shepherd, denoted by $y_d(t)\in \mathbb{R}^2$, as 
\begin{equation}
    y_d = K_1x_T,
\end{equation}
where $K_1$ is a positive constant. {This choice is intuitive because if the shepherd is always at the desired position specified above, the target sheep will be regulated to the origin by the repulsive force of the shepherd.} Then, we quantify 
the mismatch between the actual position $y$ and desired position of the shepherd as 
\begin{equation}\label{ey}
    e_y = y_d-y.
\end{equation}
In this study, we consider the following first-order dynamics for the shepherd:
\begin{equation}\label{u_y}
    \dot y = u_y.
\end{equation}

The following theorem shows that we can construct a feedback controller to \add{locally} achieve the control objective stated in Problem~\ref{prb:}. 

\begin{theorem}\label{Lyapunov}
Assume that inequality~\eqref{M_b-m_a>fy} is satisfied. 
Define $\gamma\colon \mathbb R^2 \times  \mathbb R^2\to \{\gamma_1, \gamma_2\}$ by 
\begin{equation}\label{}
\begin{multlined}
\gamma(x_T, y) = 
    \begin{cases}
        \gamma_1,&\mbox{\textup{if} }\norm{x_T-y}>\ell_y, \\
\gamma_2,&\mbox{\textup{if} }\norm{x_T-y}\leq\ell_y.
    \end{cases}
    \end{multlined}
\end{equation}
and construct a feedback controller as 
\begin{equation}\label{uy}
\begin{aligned}
u_y &= \gamma\frac{e_y}{\norm{e_y}}(K_1+\norm{x_T})\norm{w_T}\\
&+\gamma\frac{e_y}{\norm{e_y}}\left(\frac{\norm{x_T}^2+\norm{x_T}\norm{y}+K_1\norm{x_T-y}}{\norm{x_T-y}^2}\right).
\end{aligned}
\end{equation}
Define
\begin{equation}\label{V>0}
V = \frac{1}{2}\lVert x_T \rVert^2 + \norm{e_y}.
\end{equation}
\add{Then, for any constant~$c>0$, set $V\leq c$ 
is forward invariant under feedback controller~\eqref{uy}.}
\del{ensures that the target sheep is regulated to the origin in the sense that $x_T(t)\to 0$ as $t\xrightarrow[]{}\infty$.}
\end{theorem}

\begin{proof}
By inequality~\eqref{M_b-m_a>fy} and Theorem~\ref{noncollision}, the existence of the solution of the system is guaranteed. Therefore, it is sufficient to demonstrate that the time derivative 
\begin{equation}\label{dotv<0}
\begin{aligned}
\dot{V} &= x_T^{\top}\dot{x}_T + \frac{(y_d-y)^{\top}}{\norm{y_d-y}}(\dot{y}_d-\dot y)
\end{aligned}
\end{equation}
is negative. First, if $\norm{x_T-y}>\ell_y$, then \begin{equation}\label{dotv<0 case1}
\begin{multlined}
\dot{V} = \gamma_1x_T^{\top}\left(\frac{x_T-y}{\norm{x_T-y}^2}+w_T\right)\\
+ K_1\gamma_1\frac{e_y^{\top}}{\norm{e_y}}\left(\frac{x_T-y}{\norm{x_T-y}^2}+w_T\right)-\frac{e_y^{\top}}{\norm{e_y}}\dot y. 
\end{multlined}
\end{equation}
Therefore, using~\eqref{uy} and~\eqref{dotv<0 case1}, we can easily show $\dot V<0$. 
Let us consider the case where $\norm{x_T-y}\leq\ell_y$. In this case, we show that
\begin{equation}\label{dotv<0 case2}
\begin{aligned}
&\dot{V} = \gamma_2x_T^{\top}\left(\frac{x_T-y}{\norm{x_T-y}^2}g(\norm{x_T-y})+w_T\right)\\
&+ K_1\gamma_2\frac{e_y^{\top}}{\norm{e_y}}\left(\frac{x_T-y}{\norm{x_T-y}^2}g(\norm{x_T-y})+w_T\right)-\frac{e_y^{\top}}{\norm{e_y}}\dot y.
\end{aligned}
\end{equation}
Because 
    $g(\norm{x_T-y})=\exp(-{1}/{\norm{x_T-y}})<1$, 
we can further evaluate the derivative as 
\begin{equation}\label{iv}
\begin{aligned}
\dot{V} 
&<
\begin{multlined}[t][.85\linewidth]
    \gamma_2\left(\frac{\norm{x_T}^2+\norm{x_T}\norm{y}+K_1\norm{x_T-y}}{\norm{x_T-y}^2}\right)\\
+ \gamma_2(\norm{x_T}+K_1)w_T-\frac{e_y^{\top}}{\norm{e_y}}\dot y=0
\end{multlined}\\
\end{aligned}
\end{equation}
\aiyi{The formulas are maybe redundant.}\masaki{Agreed}
as desired. This completes the proof of the theorem.
\del{This completes the proof of Theorem~\ref{Lyapunov}.}
\end{proof}
\begin{figure*}[tb]
  \centering
  \begin{minipage}[t]{0.45\textwidth}
    \centering
    \subfigure[The proposed method with $N = 10$ sheep]{\label{fig:1a}\includegraphics[width=0.85\textwidth]{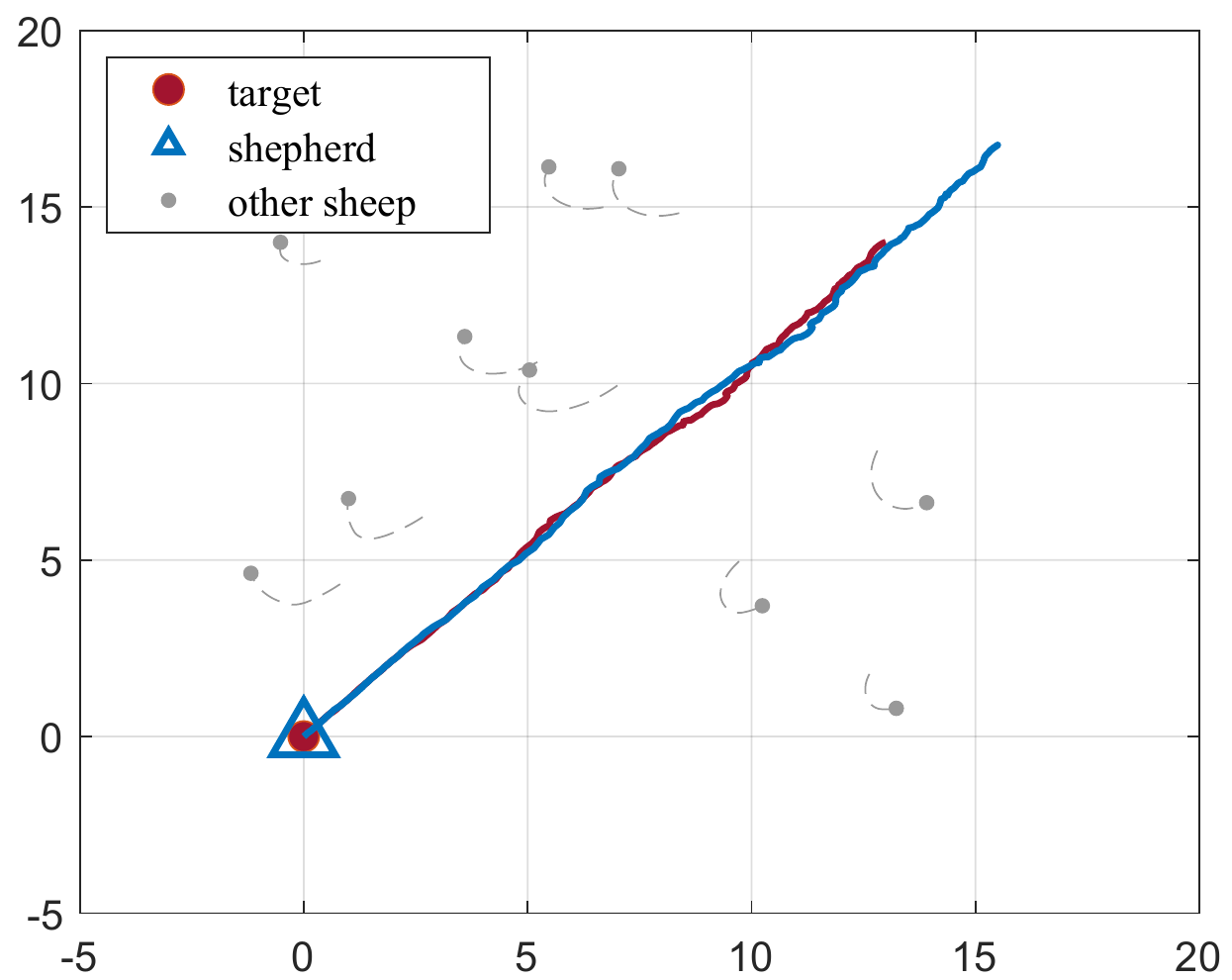}}
\end{minipage}\!\!\!\!\!\!\!\!\!\!\!\!\!
  \begin{minipage}[t]{0.45\textwidth}
    \centering
    \subfigure[The baseline method with $N = 10$ sheep]{\label{fig:1b}\includegraphics[width=0.85\textwidth]{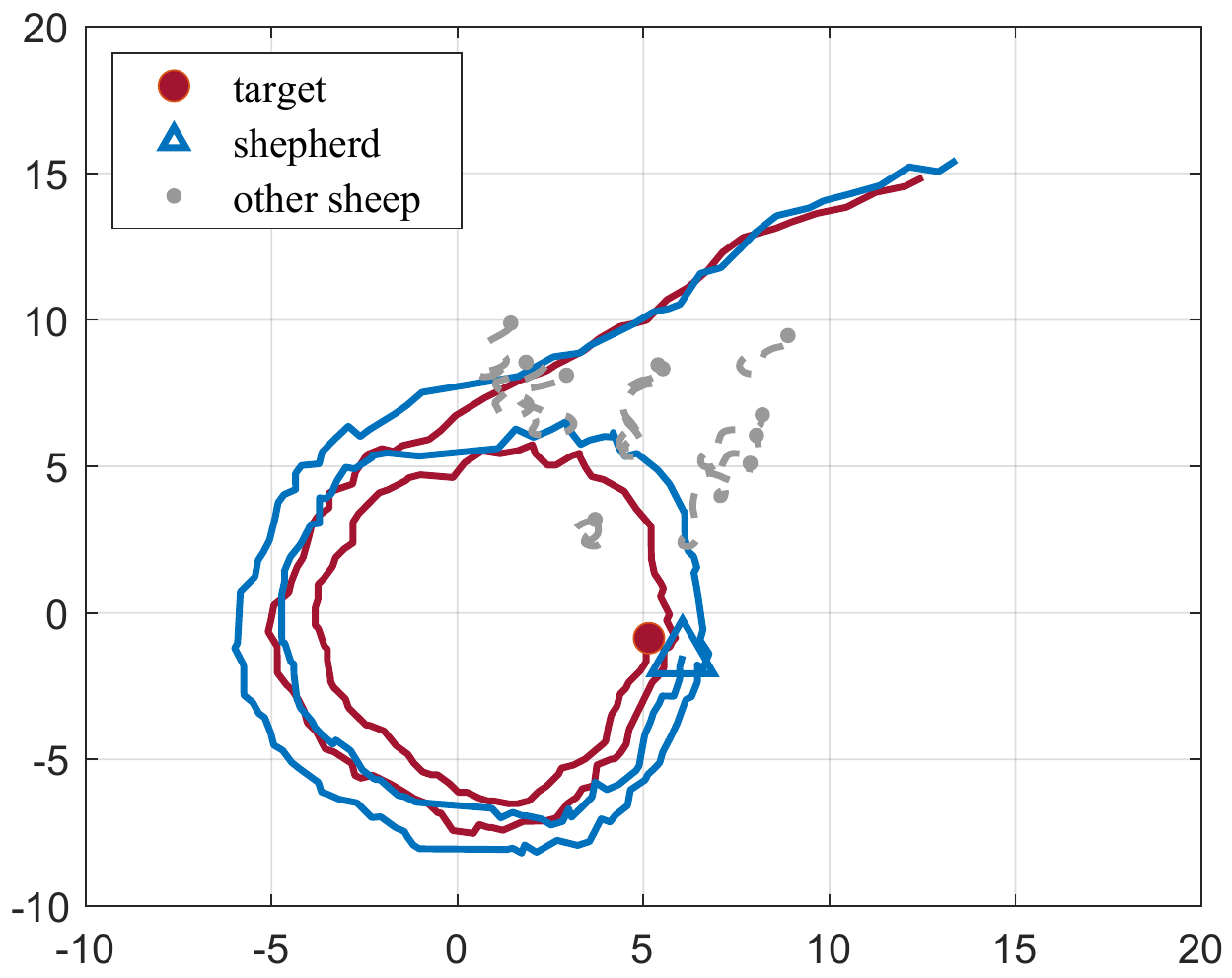}}
  \end{minipage} \\[-1pt]    %
  \begin{minipage}[t]{0.45\textwidth}
    \centering
    \subfigure[The proposed method with $N = 50$ sheep]{\label{fig:1c}\includegraphics[width=0.85\textwidth]{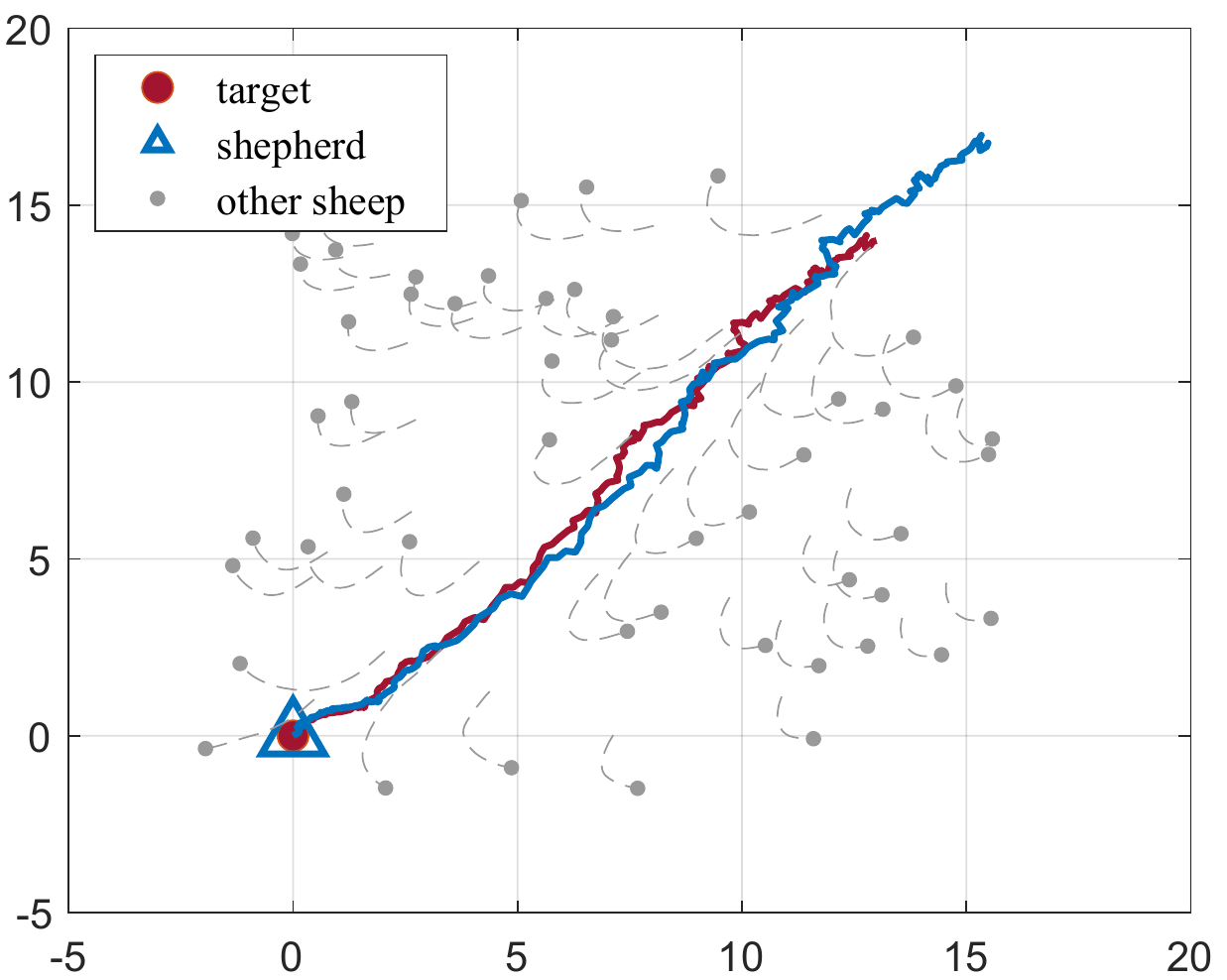}}
  \end{minipage}\!\!\!\!\!\!\!\!\!\!\!\!\!
  \begin{minipage}[t]{0.45\textwidth}
    \centering
    \subfigure[The baseline method with $N = 50$ sheep]{\label{fig:1d}\includegraphics[width=0.85\textwidth]{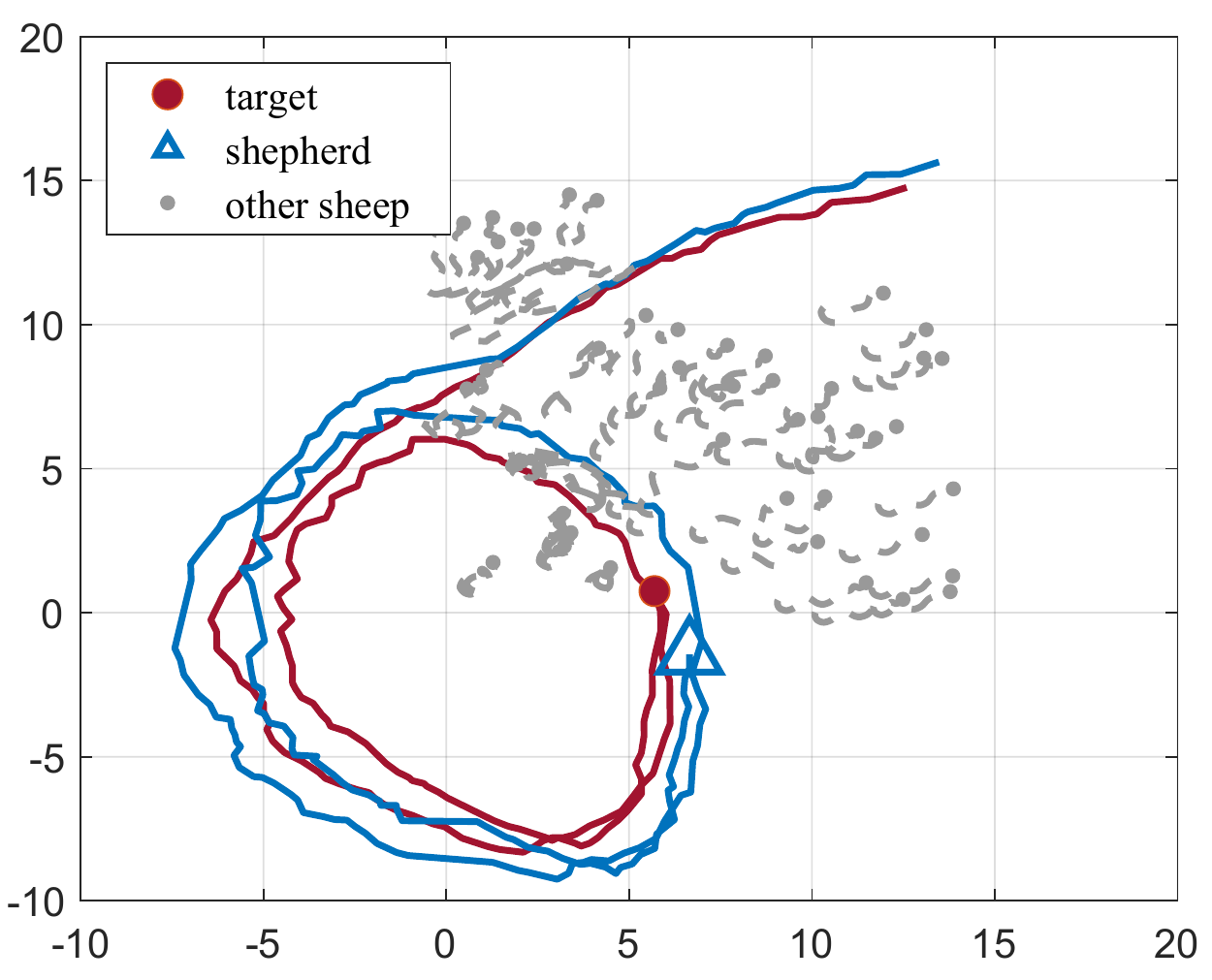}}
  \end{minipage}
  \caption{{The comparison between the proposed method and baseline. These four figures show the overall trajectory of one shepherd regulating the target to the origin in different swarm sizes, where the red line, blue line, and dashed lines represent the trajectories of the target, shepherd, and the other sheep in the swarm.}}
  \label{fig:trajectory}
\end{figure*}
\begin{remark}
    \add{Although Theorem~\ref{Lyapunov} proves only a local effectiveness of the controller~\eqref{uy}, numerical simulations presented in the next section suggests its global effectiveness. We leave the problem of theoretically establishing global efficacy of the controller~\eqref{uy} an open problem.}
\end{remark}

\section{Numerical Simulations}\label{Numerical Simulation}
In this section, we demonstrate the effectiveness of the proposed feedback controller in regulating one target sheep to the origin in a large-scale swarm. The initial conditions for the first experiment were set to $x_T(0)=[13, 14]^{\top}$ and~$y(0)=[17, 17]^{\top}$. We randomly generated $N=200$ non-overlapping sheep in the square region of~$[0,17]^2$. 
We selected the constants as $m_a = 10$, $m_b=5$, $M_b = 10$, $\ell_y = 0.5$, $\gamma_1 = 1$, $\gamma_2 = e^2$, $l=0.5$, and~$K_1 = 1.5$. These constants were selected to satisfy inequality~\eqref{M_b-m_a>fy} in Theorem~\ref{noncollision}; therefore, the swarm guarantees the collision-free property and the controller can safely regulate the target to the origin.
\begin{figure}[bt]
  \centering
\includegraphics[width=0.90\linewidth]{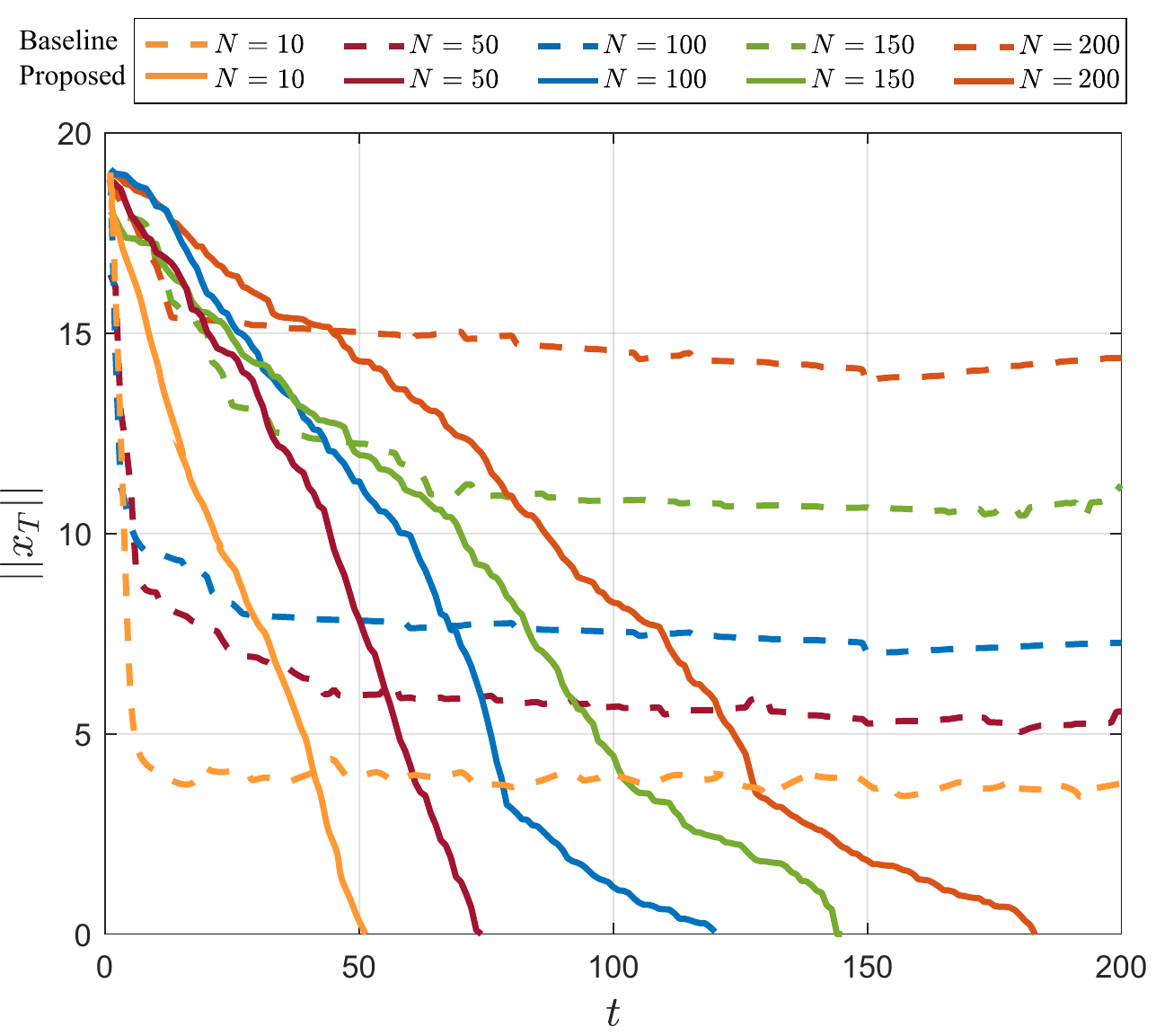}
  \caption{
  The norm of the target position~$\norm{x_T}$ varying from~\mbox{$t=0$} to~$t=200$ by baseline and proposed method, respectively.
 The dashed line represents the baseline method, while the solid line represents the proposed method.}
  \label{fig:compare}
\end{figure}
{To further illustrate the proposed method's effectiveness, we compare its performance with a baseline strategy from the farthest-agent algorithm \cite{tsunoda2018analysis}. The farthest-agent algorithm utilizes the sum of the following two vectors as the velocity input for the shepherd, where one vector drives the shepherd to approach the target, and the other vector avoids the shepherd exceeding the target.
}

{Figure~\ref{fig:trajectory} depicts two sets of comparative experiments, illustrating the change of the swarm trajectories by the proposed and baseline methods. 
{As shown in Figure~\ref{fig:trajectory}, we conducted comparative experiments to show the swarm trajectories change from $t=0$ until either $\norm{x_T}<0.1$ or a designated time limit of~$t=200$.} Figures~\ref{fig:1a} and~\ref{fig:1b} show the overall trajectories of the proposed and baseline methods on a $N=10$ swarm with the same initial placements, respectively, while Figures~\ref{fig:1c} and~\ref{fig:1d} show those on a $N=50$ swarm with the same initial placements.}
{We observe that the baseline method guides the target around the origin but fails in making the target's position converge to the origin. On the other hand, the proposed method is guaranteed to make the tracking error converge to zero and, therefore, succeeds in guiding the target to the origin asymptotically.}
{We further illustrate the effectiveness of the proposed method in different swarm sizes (from $N=10$ to $N=200$) as shown in Figure~\ref{fig:compare}. The dashed line represents the baseline method, and the solid line represents the proposed method. Each plot is obtained by averaging 100 different simulations with different initial placements. From Figure~\ref{fig:compare}, we observe the proposed method regulates the target to the origin before $t=200$ for all different swarm sizes and performs better than the baseline. This further corroborates the conclusion of the comparative analysis of the two methods mentioned above.}

\section{Conclusion and Future works}
In this study, we investigated the problem of using a shepherd to accurately guide a target in a large swarm to its destination. We first analyzed the properties of the swarm and rigorously proved through mathematical analysis that under certain inequality constraints, no sheep in the swarm will collide. Subsequently, we used the Lyapunov stability principle to design a shepherd control law that precisely guides the target to the origin. {We also compared our model with a baseline shepherding control algorithm. Our experimental results demonstrate that the shepherd can precisely guide the target to the goal in both small and large swarms.} 
Our future work includes finding necessary and sufficient conditions of the collision-free property of shepherding swarm.

\balance

\bibliographystyle{IEEEtran}
\bibliography{main}

\begin{thebibliography}{10}
\providecommand{\url}[1]{#1}
\csname url@samestyle\endcsname
\providecommand{\newblock}{\relax}
\providecommand{\bibinfo}[2]{#2}
\providecommand{\BIBentrySTDinterwordspacing}{\spaceskip=0pt\relax}
\providecommand{\BIBentryALTinterwordstretchfactor}{4}
\providecommand{\BIBentryALTinterwordspacing}{\spaceskip=\fontdimen2\font plus
\BIBentryALTinterwordstretchfactor\fontdimen3\font minus
  \fontdimen4\font\relax}
\providecommand{\BIBforeignlanguage}[2]{{%
\expandafter\ifx\csname l@#1\endcsname\relax
\typeout{** WARNING: IEEEtran.bst: No hyphenation pattern has been}%
\typeout{** loaded for the language `#1'. Using the pattern for}%
\typeout{** the default language instead.}%
\else
\language=\csname l@#1\endcsname
\fi
#2}}
\providecommand{\BIBdecl}{\relax}
\BIBdecl

\bibitem{budrene1995dynamics}
E.~O. Budrene and H.~C. Berg, ``Dynamics of formation of symmetrical patterns
  by chemotactic bacteria,'' \emph{Nature}, vol. 376, pp. 49--53, 1995.

\bibitem{vicsek1995novel}
T.~Vicsek, A.~Czir{\'o}k, E.~Ben-Jacob, I.~Cohen, and O.~Shochet, ``Novel type
  of phase transition in a system of self-driven particles,'' \emph{Physical
  Review Letters}, vol.~75, no.~6, p. 1226, 1995.

\bibitem{Gazi}
V.~Gazi and K.~M. Passino, ``A class of attractions/repulsion functions for
  stable swarm aggregations,'' \emph{International Journal of Control},
  vol.~77, no.~18, pp. 1567--1579, 2004.

\bibitem{goel}
R.~Goel, J.~Lewis, M.~A. Goodrich, and P.~Sujit, ``Leader and predator based
  swarm steering for multiple tasks,'' in \emph{2019 IEEE International
  Conference on Systems, Man and Cybernetics}, 2019, pp. 3791--3798.

\bibitem{intro-unbounded}
V.~S. Chipade and D.~Panagou, ``Multi-swarm herding: Protecting against
  adversarial swarms,'' in \emph{59th IEEE Conference on Decision and Control},
  2020, pp. 5374--5379.

\bibitem{long2020comprehensive}
N.~K. Long, K.~Sammut, D.~Sgarioto, M.~Garratt, and H.~A. Abbass, ``A
  comprehensive review of shepherding as a bio-inspired swarm-robotics guidance
  approach,'' \emph{IEEE Transactions on Emerging Topics in Computational
  Intelligence}, vol.~4, no.~4, pp. 523--537, 2020.

\bibitem{shepherding-swarm2}
J.~Hu, A.~E. Turgut, T.~Krajn{\'\i}k, B.~Lennox, and F.~Arvin,
  ``Occlusion-based coordination protocol design for autonomous robotic
  shepherding tasks,'' \emph{IEEE Transactions on Cognitive and Developmental
  Systems}, vol.~14, no.~1, pp. 126--135, 2020.

\bibitem{shepherding-swarm3}
Y.~Tsunoda, Y.~Sueoka, and K.~Osuka, ``On statistical analysis for shepherd
  guidance system,'' in \emph{2017 IEEE International Conference on Robotics
  and Biomimetics}, 2017, pp. 1246--1251.

\bibitem{shepherding-swarm4}
A.~Garrell and A.~Sanfeliu, ``Local optimization of cooperative robot movements
  for guiding and regrouping people in a guiding mission,'' in \emph{2010
  IEEE/RSJ International Conference on Intelligent Robots and Systems}, 2010,
  pp. 3294--3299.

\bibitem{engineeringvalue1}
H.~Xie, M.~Sun, X.~Fan, Z.~Lin, W.~Chen, L.~Wang, L.~Dong, and Q.~He,
  ``Reconfigurable magnetic microrobot swarm: Multimode transformation,
  locomotion, and manipulation,'' \emph{Science Robotics}, vol.~4, no.~28,
  2019.

\bibitem{engineeringvalue2}
J.~Yu, B.~Wang, X.~Du, Q.~Wang, and L.~Zhang, ``Ultra-extensible ribbon-like
  magnetic microswarm,'' \emph{Nature Communications}, vol.~9, no.~1, p. 3260,
  2018.

\bibitem{Deptula}
P.~Deptula, Z.~I. Bell, F.~M. Zegers, R.~A. Licitra, and W.~E. Dixon,
  ``Approximate optimal influence over an agent through an uncertain
  interaction dynamic,'' \emph{Automatica}, vol. 134, p. 109913, 2021.

\bibitem{Licitra}
R.~A. Licitra, Z.~I. Bell, and W.~E. Dixon, ``Single-agent indirect herding of
  multiple targets with uncertain dynamics,'' \emph{IEEE Transactions on
  Robotics}, vol.~35, no.~4, pp. 847--860, 2019.

\bibitem{onlyone1}
E.~Sebasti{\'a}n and E.~Montijano, ``Multi-robot implicit control of herds,''
  in \emph{2021 IEEE International Conference on Robotics and Automation},
  2021, pp. 1601--1607.

\bibitem{onlyone2}
D.~M. Le, X.~Luo, L.~J. Bridgeman, M.~M. Zavlanos, and W.~E. Dixon,
  ``Single-agent indirect herding of multiple targets using metric temporal
  logic switching,'' in \emph{59th IEEE Conference on Decision and Control},
  2020, pp. 1398--1403.

\bibitem{onlyone3}
W.~Zhang, V.~S. Chipade, and D.~Panagou, ``Herding an adversarial swarm in
  three-dimensional spaces,'' in \emph{2021 American Control Conference}.\hskip
  1em plus 0.5em minus 0.4em\relax IEEE, 2021, pp. 4722--4728.

\bibitem{collision5}
S.~Mayya, G.~Notomista, D.~Shell, S.~Hutchinson, and M.~Egerstedt,
  ``Non-uniform robot densities in vibration driven swarms using phase
  separation theory,'' in \emph{2019 IEEE/RSJ International Conference on
  Intelligent Robots and Systems}, 2019, pp. 4106--4112.

\bibitem{collision_otherswarm1}
X.~Yin, D.~Yue, and Z.~Chen, ``Asymptotic behavior and collision avoidance in
  the cucker--smale model,'' \emph{IEEE Transactions on Automatic Control},
  vol.~65, no.~7, pp. 3112--3119, 2019.

\bibitem{bandyopadhyay2021detection}
S.~Bandyopadhyay, V.~Gehlot, M.~Balas, D.~S. Bayard, and M.~B. Quadrelli,
  ``Detection of transient instabilities in multi-agent systems and swarms,''
  in \emph{2021 American Control Conference}.\hskip 1em plus 0.5em minus
  0.4em\relax IEEE, 2021, pp. 1216--1223.

\bibitem{collision_otherswarm3}
H.-T. Zhang, C.~Zhai, and Z.~Chen, ``A general alignment repulsion algorithm
  for flocking of multi-agent systems,'' \emph{IEEE Transactions on Automatic
  Control}, vol.~56, no.~2, pp. 430--435, 2010.

\bibitem{collision_otherswarm4}
B.~Zhang and H.~P. Gavin, ``Natural deadlock resolution for multi-agent
  multi-swarm navigation,'' in \emph{60th IEEE Conference on Decision and
  Control}, 2021, pp. 5958--5963.

\bibitem{collision_otherswarm2}
H.~Su, X.~Wang, and Z.~Lin, ``Flocking of multi-agents with a virtual leader,''
  \emph{IEEE Transactions on Automatic Control}, vol.~54, no.~2, pp. 293--307,
  2009.

\bibitem{tsunoda2018analysis}
Y.~Tsunoda, Y.~Sueoka, Y.~Sato, and K.~Osuka, ``Analysis of local-camera-based
  shepherding navigation,'' \emph{Advanced Robotics}, vol.~32, no.~23, pp.
  1217--1228, 2018.

\bibitem{syha}
S.-Y. Ha, J.~Jung, J.~Kim, J.~Park, and X.~Zhang, ``Emergent behaviors of the
  swarmalator model for position-phase aggregation,'' \emph{Mathematical Models
  and Methods in Applied Sciences}, vol.~29, no.~12, pp. 2225--2269, 2019.

\end{thebibliography}
\balance

\end{document}